\documentclass[letterpaper, 10pt, conference]{ieeeconf} 
\IEEEoverridecommandlockouts                              %
\usepackage[utf8]{inputenc}
\usepackage{graphicx}
\usepackage{amsmath}
\usepackage{amsfonts}
\usepackage{amssymb}
\usepackage{tikz}
\usepackage{circuitikz}
\usepackage{psfrag}
\usetikzlibrary{positioning}

\usepackage{amsthm}
\newtheorem{proposition}{Proposition}
\theoremstyle{definition}
\newtheorem{definition}{Definition}
\newtheorem{example}{Example}
\newtheorem{assumption}{Assumption}

\newcommand{\setreal}{\mathbb{R}}
\newcommand{\setnat}{\mathbb{N}}

\newcommand{\setrealp}{\setreal_+}

\newcommand{\linf}{\mathcal{L}_\infty}

\newcommand{\inpset}{\mathcal{U}}

\newcommand{\Iinf}{i_{\infty}}

\newcommand{\noii}{e_i}
\newcommand{\noiv}{e_v}

\newcommand{\inp}{i}
\newcommand{\vref}{v_\text{r}}

\newcommand{\Iinfh}{\hat{\imath}_\infty}

\newcommand{\vbar}{\bar{v}}
\newcommand{\vrbar}{\bar{v}_\text{r}}

\newcommand{\Gcl}{G_{c\ell}}
\newcommand{\Gclk}{G_{c\ell,k}}

\tikzset{terminal/.style=[node distance=1.25cm,
		fblock/.style={rectangle,minimum 
									size=1.0cm,thick,draw=black},
		sum/.style={circle,minimum size=0.5cm,thick,draw=black},
		point/.style={rectangle,inner sep=0pt,minimum size=0pt},
		every new ->/.style={-latex},
		every new --/.style={},
		hvpath/.style={to path={-| (\tikztotarget)}},
		vhpath/.style={to path={|- (\tikztotarget)}}
		}

\graphicspath{{./Figures/}}

\title{\LARGE \bf
Feedback for nonlinear system identification}

\author{Thiago Burghi$^{1}$, Maarten Schoukens$^{2}$ and
Rodolphe Sepulchre$^{1}$
\thanks{*The research leading to these results has received 
funding from the European Research Council under the
Advanced ERC Grant Agreement Switchlet n.670645 and from
the Brazilian federal agency for the Coordination of
Improvement of Higher Education Personnel (CAPES).}
\thanks{$^{1}$T. Burghi ({\tt\small tbb29@cam.ac.uk}) and R. 
Sepulchre ({\tt\small r.sepulchre@eng.cam.ac.uk}) are with
the Department of Engineering, Control Group, University of
Cambridge, Cambridge CB2 1PZ, UK. $^{2}$M. Schoukens
(\texttt{m.schoukens@tue.nl}) is with the
Department of Electrical Engineering,  
Eindhoven University of Technology, 5612 AZ Eindhoven,
Netherlands.}
\thanks{\copyright  EUCA. Published in: 18th European
Control Conference (ECC), Napoli, Italy, June 25-28 2019.}}

\begin{document}

\maketitle
\thispagestyle{empty}
\pagestyle{empty}

\begin{abstract}

Motivated by neuronal models from neuroscience, we consider
the system identification of simple feedback structures 
whose behaviors include nonlinear phenomena such as
excitability, limit-cycles and chaos.
We show that output feedback is sufficient to solve the
identification problem in a two-step procedure. First, the
nonlinear static characteristic of the system is extracted,
and second, using a feedback linearizing law, a mildly
nonlinear system with an approximately-finite memory is
identified. In an ideal setting, the second step boils down
to the identification of a LTI system. To illustrate  
the method in a realistic setting, we present numerical 
simulations of the identification of two classical 
systems that fit the assumed model structure.
\end{abstract}
\begin{keywords}
	Excitability, Approximately-finite memory, Systems
	identification, Nonlinear systems, Output feedback 
\end{keywords}

\section{Introduction}


System identification of nonlinear dynamical systems has
been a topic of increasing interest in the recent years, 
see e.g. \cite{schoukens_parametric_2012} 
\cite{schoukens_identification_2017}. The approach in 
these references is block-oriented, and finds its roots in
specific structures such as Wiener-Hammerstein
models \cite{ljung_system_1999}. These block-oriented
approaches exploit the idea of estimating a best linear
approximation \cite{schoukens_identification_2005} of the
nonlinear system as a first step in the direction of solving the 
identification problem. A common underlying assumption in
the estimation of approximate linear models is that the
system class has some variant of the fading memory property, 
meaning that the output signals depend on the past of the 
input signals with a forgetting factor, see e.g. 
\cite{boyd_fading_1985} and 
\cite{sandberg_approximately-finite_1992}. 

The present work seeks to extend the above methods to
input-output nonlinear behaviors that can be transformed 
by output feedback into operators with a fading memory. 
More specifically, we observe that the simple
interconnection structure in Figure \ref{fig:lure_passive}
possesses that property, by inspection, and is general
enough to include nonlinear behaviors that are hard to
identify with state-of-the art methods.

In particular, we are motivated by conductance-based models 
of neurons. Those models, pioneered by Hodgkin and Huxley 
in their seminal work \cite{hodgkin_quantitative_1952}, have 
become central to neurophysiology and computational neuroscience.
Their behaviors include nonlinear phenomena such as
excitability, limit cycles, bistability, and bursting. Yet,
all conductance-based models share the structure in Figure
\ref{fig:lure_passive}, where the passive element models the
passive behavior of the cellular membrane and the 
fading memory operator models the voltage-gated conductance
of ion channels. We advocate that such models can be transformed 
by feedback into operators with a fading memory, and that this
property makes them amenable to rigorous system identification. 
This property is in fact at the root of the voltage-clamp 
experiment that has been central to the conductance-based 
modelling principle over the last seventy years.

As a first step, in this paper, we focus on the elementary
situation where the fading memory component in Figure
\ref{fig:lure_passive} is static, and the passive element is LTI.
The feedback structure then becomes the classical structure 
of a Lure system. This simplified structure already 
includes famous models such as the excitable circuit of 
Fitzugh and Nagumo \cite{nagumo_active_1962} and the 
chaotic circuit of Chua \cite{matsumoto_chaotic_1984}. 
We show that the identification of such nonlinear circuits 
becomes straightforward if we introduce output feedback in
experiment design. Not surprisingly, the static element can be
identified separately from the LTI element. This allows the use
of a feedback linearizing law to transform the identification 
problem into that of identifying a mildly nonlinear 
system with an approximately-finite memory 
\cite{sandberg_approximately-finite_1992} -- a specific
type of fading memory property.


\begin{figure}
	\centering
	\begin{tikzpicture}
		\node (v_dyn) at (0,0) [fblock,label={Passive}]{$G$};		
		\node (int_dyn) [fblock, below= 0.3 of v_dyn,
		align=center]{Fading \\ memory};
		\node (sum1) [sum,left=of v_dyn,
						label={175:$+$},label={-85:$-$}] {};		
		\coordinate[left= 1 of sum1] (input);
		\coordinate[right=of v_dyn] (split1);	
		\coordinate[right=of split1] (output);		
		\draw[-latex] (sum1) -- (v_dyn);
		\draw[-latex] (input) -- (sum1) node[left,pos=0]{$\inp(t)$};
		\draw[-latex] (v_dyn) -- (split1) -- (output) 
		node[right,pos=1]{$v(t)$};
		\draw[-latex] (split1) |- (int_dyn);
		\draw[-latex] (int_dyn) -| (sum1);
	\end{tikzpicture}
	\caption{A nonlinear feedback circuit. In this paper, the fading memory block 
	is a static nonlinearity $h(\cdot)$, and $G$ is LTI.}
	\label{fig:lure_passive}
\end{figure}
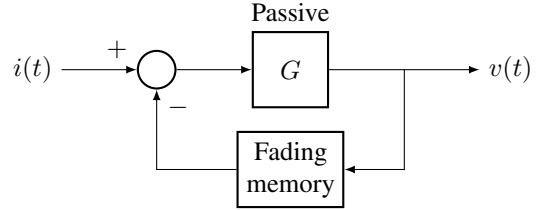

Although elementary, we believe that this methodology is
general and appealing for the identification of nonlinear
systems that do not have the fading memory property. 
This methodology is also in line with the idea that smart
experiment design is important to obtain good models of
nonlinear systems.


The paper is organized in the following way. In Section
\ref{sec:model_structure}, we define the model
class which we are interested in identifying,
and give two examples of systems that fit in that 
class. In Section \ref{sec:fm_fb}, we 
recall the concept of approximately-finite memory, and
show that output feedback can be used to endow
systems in the defined model class with that property.
In Section \ref{sec:feedlin}, we develop the main
contribution of the paper, based on a two-step
identification procedure for identifying systems
in the model class. In Section \ref{sec:numerical} we
present numerical simulations concerning the identification
of the examples from Section \ref{sec:model_structure}
in a realistic setting. Some concluding remarks are
presented in Section \ref{sec:conductance}.

\section{Model structure}
\label{sec:model_structure} 


The model and input classes of interest are defined below.

\begin{definition}[Model class]
	\label{def:model_class}
	We will work with the class of systems given by
	the negative feedback interconnection between a causal 
	LTI component $G$ and a nonlinear static map $h$, as in
	Figure \ref{fig:lure_passive}. The map 
	$h:\setreal\to\setreal$ is a continuous function such
	that, without loss of generality, $h(0) = 0$. In
	addition, there are two real constants $\rho_1$ and 
	$\rho_2$ such that 
	\begin{equation}
		\label{eq:slope_condition}
		\rho_1 \le \frac{h(v_2)-h(v_1)}{v_2-v_1}  \le \rho_2
	\end{equation}
	for all $v_2 \neq v_1$. 
	The LTI component $G$ belongs to the set of
	real-rational, strictly proper transfer
	functions $G(s) = N(s)/D(s)$ such that
	all poles of $G(s)$ are in $\mathrm{Re}[s] < 0$,
	$\mathrm{Re}[G(j\omega)] \ge 0$ for all 
	$\omega \in \setreal$, and $G(0)>0$.
\end{definition}

The above implies that $G(s)$ is positive-real
\cite[Definition 6.4]{khalil_nonlinear_2002}, and that
$\mathrm{deg}\,D(s) - \mathrm{deg}\,N(s) = 1$. We denote
$\|G\|_1 = \int_0^\infty |g(t)| dt$, where $g(t)$ is the
impulse response of $G(s)$.

\begin{definition}[Input class]
	\label{def:input_class} 
	For an arbitrary $\xi>0$, the input class 
	$\mathcal{U} \subset \linf(\setrealp)$ is the set of
	functions $u$ from $\setrealp = [0,\infty)$ to 
	$\setreal$ such that $\sup_{t\ge0} |u(t)| < \xi$.
\end{definition}


\subsection{Some examples}
\label{sec:examples} 

In this section, we provide two simple examples of
circuits that belong to the model class defined
above. 

\begin{example}
\label{ex:fhn} 
The Fitzhugh-Nagumo (FHN) circuit \cite{nagumo_active_1962} 
was proposed as a simple model of realistic neurons and
became a paradigm of excitability. The model has the 
state-space representation
\begin{equation}
	\label{eq:gen_vdp_ss}
	\begin{split}
		\tfrac{1}{20} \dot{v} &= - x - h(v) + \inp \\
		\dot{x} &= - \tfrac{3}{4} x + v 
	\end{split}
\end{equation}
where $h$ is given by the nonlinear characteristic\footnote{
Note that if $i \in \inpset$ we can always choose a bounded 
positively invariant state-space $X$ for this system. 
Then, $h(v)$ satisfies \eqref{eq:slope_condition} in $X$. Such 
a set $X$ can be found, for instance, using the Lyapunov function
$V(v,x) = v^2/2 + 10x^2$ and the standard arguments in
\cite[Section 4.8]{khalil_nonlinear_2002}.}
\begin{equation}
	\label{eq:fhn_nl}
		h(v) = - v + v^3/3
%
\end{equation}

Note that $\rho_1 = -1$, and the nonlinear resistance 
is locally active. It can be verified that the system 
\eqref{eq:gen_vdp_ss}-\eqref{eq:fhn_nl} belongs to the 
model class of Definition \ref{def:model_class}, with
\begin{equation}
	\label{eq:fhn_l}
	 G(s) = \frac{20s + 15}{s^2 + 0.75s + 20}.
\end{equation}
For $i=0$, the system behaves as an autonomous relaxation
oscillator.
For $i=-1.5$, the output $v(t)$ converges to a
constant equilibrium, and the system is \textit{excitable}:
the output can display high-amplitude excursions
away from equilibrium, called \textit{spikes}, when the
input $i$ is increased momentarily past a certain
excitability threshold \cite{sepulchre_excitable_2018}. 


\end{example}

\begin{example}
\label{ex:chua} 
	The Chua circuit \cite{matsumoto_chaotic_1984} is	
	constructed with two capacitors $c_1>0$ and $c_2>0$, an
	inductor $\ell>0$, a resistor $r>0$ and a Chua diode.
	The Chua diode is a nonlinear resistive element with a
	piecewise-linear monotonically decreasing characteristic
	given by
	\begin{equation}
	\label{eq:chua_diode} 
		h(v) = \left\{ \begin{array}{cc}
		-0.1(v+1) + 4, & 
		v\le -1 \\ 
		-4v, & -1<v<1 \\ 
		-0.1(v-1) -4 & 
		v \ge 1
		\end{array}   \right.
	\end{equation}
	
	The passive element of the Chua circuit is given by
	\[ G(s) = \frac{\ell c_2 s^2 + \ell r s + 1}
	{\ell c_1 c_2 s^3+ \ell r (c_1+c_2) s^2 + c_1 s + r} \]
	
	In \cite{matsumoto_chaotic_1984}, it is shown that the
	autonomous Chua circuit presents chaotic behavior when
	the	parameters are given by $c_1 = 0.1$, $c_2 = 2$,
	$\ell = 1/7$ and $r=0.7$. 
	By forcing the Chua circuit with an external current,
	the circuit belongs to the model class of Definition 
	\ref{def:model_class}. Note that in this
	case $\rho_1 = -4$ and $\rho_2 = -0.1$. 
\end{example}

\section{Approximately-finite memory through 
output feedback}
\label{sec:fm_fb} 

In this section, we discuss how the feedback law
\begin{equation}
	\label{eq:linear_fb} 
	\inp = k(\vref - v)
\end{equation}
is used to endow a system from the model class of 
Definition \ref{def:model_class} with the 
approximately-finite memory property 
\cite{sandberg_approximately_1994}.

\subsection{Approximately-finite memory}
\label{sec:fading_memory} 

Consider the model class of Definition \ref{def:model_class} 
and the input class of Definition \ref{def:input_class}. 
Let $G$ denote the (convolution) operator defined by $G(s)$ 
and $H$ denote the operator defined by $(Hv)(t) = h(v(t))$. 
It can be shown\footnote{See e.g. 
\cite[Section 2.3]{sandberg_approximately_1994}, where 
$(I+GH)^{-1}$ is denoted by $V$.}, based on the stability of 
$G(s)$ and the Lipschitz property of $h$, that the map 
$(I+GH)^{-1}$ is well defined on $\linf(\setrealp)$. Thus,
\begin{equation*}
	\begin{split}
		v &= (I+GH)^{-1}(Gi+g_0)
	\end{split}
\end{equation*}
where $g_0 \in  \linf(\setrealp)$ is a term taking into 
account the exponentially decaying initial conditions
of the linear system. 

Let $\Gcl$ denote the restriction of $(I+GH)^{-1}G$
to $\inpset$ (under zero initial conditions, 
this is the map from the input $i$ to the output $v$). 
We are interested in the following property.

\begin{definition}[\cite{sandberg_approximately_1994}]
\label{def:approximately_finite} 

	Let $F:\inpset \to \linf(\setrealp)$ be a causal 
	time-invariant operator. We say $F$ has 
	approximately-finite memory with respect to 
	$\inpset$, or $F \in \mathcal{A}(\inpset)$, if for any
	given $\epsilon>0$, there is a $\Delta>0$ such that 
	\begin{equation}
		\label{eq:fading_memory} 
		|(F u)(t) - (FW_{t,\eta}u)(t)| < \epsilon, 
	\quad t\ge 0	
	\end{equation}	
	for all $\eta\ge\Delta$ and all $u \in \inpset$, 
	where $W_{t,\eta}$ is the window operator
	\begin{equation}
		(W_{t,\eta}u)(\tau) = \left\{ 
			\begin{array}{cc}
					u(\tau), &\quad t-\eta \le \tau \le t \\ 
					0, &\text{otherwise}
					\end{array} 		
		\right.		 
	\end{equation}
\end{definition}

The inequality \eqref{eq:fading_memory} shows that 
the recent past of the input of a system in 
$\mathcal{A}(\inpset)$ dominates the behavior of its
output. An important result linking Definition
\ref{def:approximately_finite} to the circle criterion is
\cite{sandberg_approximately_1994}. In our context, we have
the following statement.

\begin{proposition}
\label{prop:circle_criterion} 
Assume that one of the following two conditions are satisfied:

(i) $0 \le \rho_1 < \rho_2$, all poles of $G(s)$ are 
	in $\mathrm{Re}[s] < 0$, and 
	$\mathrm{Re}[G(j\omega)]\ge 0$ for all 
	$\omega \in \setreal$. 

(ii) $\rho_1 < 0 < \rho_2$, all poles of $G(s)$ are 
	in $\mathrm{Re}[s] < 0$, and the locus of 
	$G(j\omega)$ for $-\infty<\omega<\infty$ is 
	contained within the circle of radius 
	$(\rho_2^{-1}-\rho_1^{-1})/2$ centered  
	on the real axis of the complex plane at 
	$-(\rho_2^{-1}+\rho_1^{-1})/2 + j0$.

Then $\Gcl$ has approximately-finite memory on $\inpset$.
\end{proposition}
\begin{proof}
	Let $\inpset'$ be defined similarly to $\inpset$, but 
	with $\xi' = \|G\|_1 \xi$. Under our
	assumptions, \cite[Theorem 1]{sandberg_approximately_1994} 
	ensures that the map $GH(I+GH)^{-1} \in 
	\mathcal{A}(\inpset')$ (for simplicity, we denote operators
	and their restrictions by the same symbols).
	But since $GH(I+GH)^{-1} = I - (I+GH)^{-1}$, it 
	follows from direct application of the inequality 
	\eqref{eq:fading_memory} that $(I+GH)^{-1}$ is also in
	$\mathcal{A}(\inpset')$.
	Thus, $\Gcl$ is the cascade interconnection of 
	$(I+GH)^{-1}\in \mathcal{A}(\inpset ')$ with 
	$G \in \mathcal{A}(\inpset)$. Since
	$G u \in \inpset '$ for all $u \in \inpset$, $\Gcl$ can 
	be shown to be in $\mathcal{A}(\inpset)$ using the
	cascade	interconnection result\footnote{This result
	requires $(I+GH)^{-1}$ to be uniformly continuous on 
	$\linf(\setrealp)$, which can be shown by means of 
	\cite[Corollary 3a]{sandberg_results_1965}.}
	\cite[Theorem 3]{sandberg_approximately-finite_1992}.
\end{proof}


\subsection{Linear output feedback}

If $h(v)$ possesses regions of negative conductance,
i.e., $\rho_1<0$, and $G(s)$ fails to satisfy the circle
condition (ii) of Proposition \ref{prop:circle_criterion},
the interconnection of Definition \ref{def:model_class}
might fail to belong to $\mathcal{A}(\inpset)$ for any 
$\inpset$. In fact, we can argue that is the case for the
two examples of Section \ref{sec:examples}. The 
Fitzhugh-Nagumo model, for instance, does not satisfy
\eqref{eq:fading_memory} for the input 
$i_p(t) = (\mu(t-t_1) - \mu(t-t_2))\xi/2$, with $\mu$ the
Heaviside function and $t_2>t_1>0$. With zero initial
conditions, this input can be used to drive the state
of \eqref{eq:gen_vdp_ss}-\eqref{eq:fhn_nl} 
away from an unstable equilibrium at the origin and towards
a stable limit-cycle. As a consequence, for any constant
$\eta$, \eqref{eq:fading_memory} cannot hold for arbitrarily 
large $t>0$ and small $\epsilon$. A similar
argument can be used for the Chua circuit, where the
limit-cycle is replaced with a chaotic attractor.


The feedback law \eqref{eq:linear_fb}
can be used to endow the closed-loop operator with the
approximately-finite memory property.
To see this, note that the closed-loop feedback system
with input $k \vref$ and output $v$ can be described
by the negative feedback interconnection of $G(s)$
with the static nonlinearity
\begin{equation}
	\label{eq:nl_fb}
		h_k(v) = h(v) + kv 
\end{equation}
so that now we have
\[
	\rho_1 + k \le \frac{h_k(v_1)-h_k(v_2)}{v_1-v_2} \le \rho_2 + k
\]
for all $v_1\neq v_2$. 



Now, it is possible to make $\rho_1 + k \ge 0$ by choosing 
$k>0$ large enough. Let $H_k$ denote the operator defined by
$(H_k v)(t) = h_k(v(t))$, and consider the new closed-loop
operator $\Gclk = (I+GH_k)^{-1}Gk$. Now (i) of Proposition  
\ref{prop:circle_criterion} is satisfied, and we have 
$\Gclk \in \mathcal{A}(\inpset)$.

\section{A feedback identification method}
\label{sec:feedlin} 

In this section, we show that it is possible to 
decouple the problem of identifying a nonlinear
system belonging to the model class of Definition 
\ref{def:model_class} into a nonlinear static 
identification stage and a dynamic mildly
nonlinear identification stage. We work with the 
following simplifying assumption.

\begin{assumption}[Simplified setup]
	\label{asp:ideal_case} 
	The model class is described by Definition 
	\ref{def:model_class}. In addition,	$h$ is given by
	\begin{equation}
	\label{eq:asp_h} 
		h(v) = a_1 v + \sum_{j=2}^{J^*} a_j \phi_j(v)
	\end{equation}
	where $a_j \in \setreal$ and
	$J^* \in \setnat \cup \{\infty\}$. We assume the
	$\phi_j$ are known linearly independent functions which
	are Lipschitz continuous on every bounded subset of 
	$\setreal$. The feedback
	law	$i=k(\vref-v)$, with $k+\rho_1 > 0$, is implemented 
	with an ideal analog circuit. The signal $v_r$ 
	is known, and the signals $i_m = i+\noii$ and $v_m = v+\noiv$ 
	are observed, where $\noii$ and $\noiv$ are independent 
	Gaussian coloured zero-mean noise terms with finite variances. 
	Figure \ref{fig:simple_setup} with the block $K = k(v_r-v)$ 
	gives a representation of this setup.
\end{assumption}

\begin{figure}
	\centering
	\begin{tikzpicture}
		\node (v_dyn) at (0,0) [fblock]{$G(s)$};		
		\node (int_dyn) [fblock, below= 0.3 of v_dyn,
		align=left]{$h(\cdot)$};
		\node (sum1) [sum,left= 0.5of v_dyn,
						label={175:$+$},label={-85:$-$}] {};
		\node (cont) [fblock, left= 1 of sum1,		
		align=left]{$K$};
		\coordinate[left= 1 of cont] (input);	
		\coordinate[right=of v_dyn] (output);		
		\coordinate[below=2 of output] (split_v);
		\coordinate[left=0.5 of sum1] (split_i);
		\node (sum3) [sum,above=0.5of output,
						label={175:$+$},label={-95:$+$}] {};
		\node (sum4) [sum,above=0.5of split_i,
						label={175:$+$},label={-95:$+$}] {};
		\coordinate[above=0.5 of sum4] (meas_i);
		\coordinate[above=0.5 of sum3] (meas_v);
		\coordinate[left=1 of sum4] (noise_i);
		\coordinate[left=1 of sum3] (noise_v);
		\draw[-latex] (input) -- (cont) 
						node[below,pos=0]{$v_r(t)$};
		\draw[-latex] (cont) -- (sum1)
						node[below,pos=0.5]{$i(t)$};		
		\draw[-latex] (sum1) -- (v_dyn);
		\draw[-latex] (v_dyn) -- (output) 
						node[right,pos=1]{$v(t)$} --
						(split_v) -| (cont.south);
		\draw[-latex] (output) |- (int_dyn);
		\draw[-latex] (int_dyn) -| (sum1);
		\draw[-latex] (output) -- (sum3);
		\draw[-latex] (sum3) -- (meas_v)
						node[right,pos=1]{$v_m(t)$};
		\draw[-latex] (split_i) -- (sum4);
		\draw[-latex] (sum4) -- (meas_i)
						node[right,pos=1]{$i_m(t)$};
		\draw[-latex] (noise_i) -- (sum4) 
						node[above,pos=0]{$\noii(t)$};
		\draw[-latex] (noise_v) -- (sum3) 
						node[above,pos=0]{$\noiv(t)$};			
	\end{tikzpicture}
	\caption{Simplified output noise setup. $K=k(v_r-v)$ in the
	static identification stage, and $K=\kappa(v_r,v)$ from 
	\eqref{eq:feedlin} the dynamic stage.}
	\label{fig:simple_setup}
\end{figure}
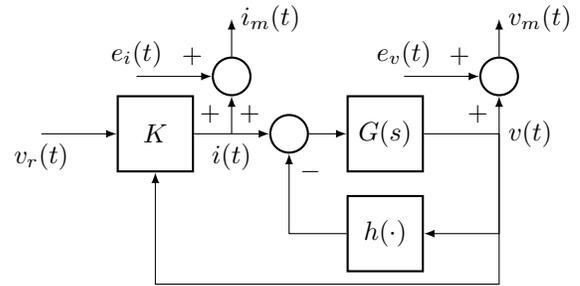

\subsection{Static identification stage}
\label{sec:ideal_nl}

We start by introducing the following concept.

\begin{definition}
	We define the inverse static input-output characteristic 
	by
	\begin{equation}
	\label{eq:inv_io} 
		\Iinf(v) = \frac{1}{G(0)} v  + h(v)	
	\end{equation}
	where $G(0)>0$ by assumption.
\end{definition}
\vspace{0.5em}

The characteristic $\Iinf(v)$ gives the (unique)
constant input required to establish an equilibrium 
at each constant $v$.
Notice that under Assumption 
\ref{asp:ideal_case}, estimating $\Iinf$ effectively 
amounts to estimating the nonlinear terms of $h$ in
\eqref{eq:asp_h}.
%

To estimate $\Iinf$, we need to stabilize the system at
different steady-states $\bar{v}$. We ensure this by means 
of the output feedback \eqref{eq:linear_fb}. The equilibrium 
of the system must satisfy
\begin{equation}
	\label{eq:fb_equilibrium} 
	-\frac{1}{G(0)} v + k \vref
	= h(v) + k v, 
\end{equation}
By assumption, the right-hand side of
\eqref{eq:fb_equilibrium} is monotonically increasing. Since 
$h$ is continuous, it follows that \eqref{eq:fb_equilibrium} 
has a single solution $\vbar$ for every $\vrbar$. The fact 
that the system settles to the unique $\vbar$ when subject 
to a constant $\vrbar$ is guaranteed by the approximately-finite
memory property \cite[Theorem 2]{sandberg_approximately-finite_1992}.
This can be alternatively be shown using the circle criterion 
\cite[Theorem 7.2]{khalil_nonlinear_2002}.

%

We can now discuss how to estimate $\Iinf$. A simple
procedure begins by choosing a sufficiently large $k>0$ and
a grid of $M$ constant values for $\vrbar$. Assume this grid
is contained in the vector $\bar{V}_\text{r}$. For each 
$m=1,\dotsc,M$, we apply the input $\bar{V}_\text{r}[m]$ to the
closed-loop system and wait for the system to settle to a
corresponding output equilibrium $\bar{V}[m]$. This yields
(as $t\to \infty$) an $M$-dimensional vector of true 
output steady-state values $\bar{V}$. In practice, the 
noise assumptions allow us to obtain consistent 
estimates $\hat{v}$ for $\vbar$ and $\hat{\imath}$ 
for $\Iinf(\vbar)$ by averaging the measurements, 
\begin{equation}
	\label{eq:meas_estimators}
	\hat{v}_N = \frac{1}{N} \sum_{n=1}^N v_m(nT_s),
	\quad 
	\hat{\imath}_N = \frac{1}{N} \sum_{n=1}^N i_m(nT_s) 
\end{equation}
where $T_s$ is the measurement sampling period and $N$ is the
number of samples. This yields estimate vectors $\hat{V}_N$
and $\hat{I}_N$.



Considering Assumption \ref{asp:ideal_case}, a natural estimator 
for $\Iinf$ is
\begin{equation}
	\label{eq:estimated_i_inf} 
	\Iinfh(v) = w_1 v + \sum_{j=2}^J w_j \phi_j(v)
\end{equation}
where $w_j$ are the estimator parameters, and $J\in\setnat$ is
such that $J \le M$. In order to estimate these parameters, we 
construct a matrix $\Phi_{N,J} \in \setreal^{M \times J}$ whose 
$m^{\text{th}}$ row is given by						
\begin{equation}
	\label{eq:regressors}
	\big(\hat{V}_N[m], \phi_2(\hat{V}_N[m]),\dotsc,
	\phi_J(\hat{V}_N[m]) \big) 
\end{equation}

Assume that $\Phi_{N,J}$ has full rank. This can be 
accomplished by choosing a sufficiently wide and fine grid
for the elements of $\bar{V}_\text{r}$. Then, a parameter
estimate $\hat{W} = (\hat{w}_1,\dotsc,\hat{w}_J)^T$ is 
obtained by solving
\begin{equation}
	\label{eq:regression} 
	\min_W \quad \sum_{m=1}^M 
	\left(	
	\hat{I}_{N}[m] - \Iinfh(\hat{V}_N[m])
	\right)^2	
\end{equation}
which yields
\begin{equation}
	\label{eq:ls_parameters} 
	\hat{W}_{N,J} = (\Phi_{N,J}^T\Phi_{N,J})^{-1}
	\Phi_{N,J}^T \hat{I}_{N}
\end{equation}

We thus have that, as $N\to \infty$ and $J \to J^*$, as long
as $\Phi_{N,J}$ has full column rank 
for all $J$,  $\Iinfh(v)$ converges to $\Iinf(v)$,
and each $\hat{w}_j$ converges to $a_j$ for $j=2,3,\dotsc,J^*$
(we drop the subscripts $N$ and $J$ of $\hat{w}_j$
for clarity).

\subsection{Dynamic identification stage}
\label{sec:ideal_lin} 

The main idea in the dynamic identification stage is to use
the input
\begin{equation}
	\label{eq:feedlin} 
	\inp = \kappa(v,\vref) \triangleq
		 k(\vref-v) + \sum_{j=2}^J \hat{w}_j \phi_j(v)
\end{equation}
so as to linearize the system by feedback. 

\begin{assumption}
\label{asp:ideal_feedback} 
The feedback law \eqref{eq:feedlin} is implemented with an 
ideal analog circuit. The setup of the problem is represented 
by Figure \ref{fig:simple_setup}, with $K = \kappa(v,\vref)$ 
given by \eqref{eq:feedlin}.
\end{assumption}

From the analysis in the previous section, as $N\to \infty$ and 
$J \to J^*$, the identification problem becomes one of identifying 
a linear system with input $\vref$, output $v$, and an output error 
structure. The ground truth model, 
at those limits, is given by $G_k(s) = k G_a(s) / ( 1 + k G_a(s))$, 
with $G_a(s) = G(s)/(1+a_1G (s))$. The system $G_a(s)$ lumps 
together
the term $a_1 v$ and the transfer function $G(s)$,
which are indistinguishable from each other from the input-output 
perspective. The resulting linear identification problem is a 
well-known one for which consistency guarantees can 
be obtained with a variety of methods \cite{ljung_system_1999}.

In practice, obviously, $N$ and $J$ will be finite, and the
nonlinearity will not be perfectly canceled by feedback. In
that case, identifying the closed-loop system from $\vref$ 
to $v$ amounts to identifying a mildly nonlinear system that 
has an approximately-finite memory and is subject to output noise.
The Best Linear Approximation (BLA) framework 
\cite{schoukens_identification_2005} ensures in this setting 
that, by using linear identification methods which are based 
on minimizing a squared sum of output residuals, we can obtain 
(asymptotically) an optimal approximation of the nonlinear system. 
Optimality, in this case, is defined with respect to the assumed 
input class \cite{schoukens_identification_2017}. Furthermore, due 
to the fact that operators with approximately-finite memory map 
periodic inputs to asymptotically periodic outputs 
\cite[Theorem 9]{sandberg_approximately-finite_1992},
by choosing periodic exciting signals, we can mitigate noise 
effects in the output by averaging the signal $v_m$ over 
different periods.


Given a best linear estimate $\hat{G}_k(s)$, to recover the 
estimate of the original nonlinear system with input $i$ and 
output $v$, we first compute 
\begin{equation}
	\label{eq:original_G} 
	\hat{G}_a(s) = 
	\frac{1}{k}\frac{\hat{G}_k(s)}{1-\hat{G}_k(s)},
\end{equation}
which is necessary to account for the $k(\vref-v)$ term in
\eqref{eq:feedlin}. The identified nonlinear system is then
given by interconnecting, in negative feedback, the transfer
function $\hat{G}_a(s)$ and the nonlinearity 
$\hat{h}(v) = \sum_{j=2}^J\hat{w}_j \phi_j(v)$.

\section{Simulations with a realistic setup}
\label{sec:numerical} 

In a more realistic identification setting, the user-defined 
feedback loop around the physical system is implemented in
discrete-time, and output measurement noise is fed 
back into the system dynamics. 

In this section, using numerical simulations, we naively 
apply the procedure described in Section \ref{sec:feedlin} 
to identify the two systems from section \ref{sec:examples}, 
assuming the realistic setup of Figure \ref{fig:real_setup}.
We assume that $\noiv$ and $\noii$ are given by white Gaussian 
noise with the same variance, denoted by $\sigma$. With this, 
we aim to provide a proof of concept that the method still 
performs well in a realistic scenario. 


\begin{figure}
	\centering
	\begin{tikzpicture}
		\node (v_dyn) at (0,0) [fblock]{$G(s)$};		
		\node (int_dyn) [fblock, below= 0.3 of v_dyn,
		align=left]{$h(\cdot)$};
		\node (sum1) [sum,left= 0.5of v_dyn,
						label={175:$+$},label={-85:$-$}] {};
		\node (cont) [fblock, left= 2.5 of sum1,		
		align=left]{$K$};
		\node (zoh) [fblock,left=0.7of sum1]{\small ZOH};
		\coordinate[left= 1 of cont] (input);	
		\coordinate[right=of v_dyn] (output);		
		\coordinate[right=0.5 of cont] (split_i);
		
		\node (sum3) [sum,below=1.75 of split_i,
						label={95:$+$},label={5:$+$}] {};
		\node (sum4) [sum,above=0.5of split_i,
						label={175:$+$},label={-95:$+$}] {};
	
		\coordinate[right=2 of sum3] (sampler1);
		\coordinate[right=1.5 of sum3] (sampler2);
		\coordinate[above=0.5 of sum4] (meas_i);
		\coordinate[left=2 of sum3] (meas_v);
		\coordinate[left=1 of sum4] (noise_i);
		\coordinate[above=0.5 of sum3] (noise_v);
		\draw[-latex] (input) -- (cont) 
						node[below,pos=0]{$\vref(nT_s)$};
		\draw[-latex] (cont) -- (zoh);
		\draw[-latex] (zoh) -- (sum1)
						node[below,pos=0.5]{$i(t)$};		
		\draw[-latex] (sum1) -- (v_dyn);
		\draw (v_dyn) -- (output) node[above,pos=0.6]{$v(t)$} 
						|- (sampler1);
		\draw (sampler1) -- ($(sampler1)+(-0.5,0.5)$);
		\draw[-latex] (sampler2) -- (sum3);
		\draw[-latex] (sum3) -| (cont);
		\draw[-latex] (sum3) -- (meas_v) node[above,pos=1]{$v_m(nT_s)$};
		\draw[-latex] (output) |- (int_dyn);
		\draw[-latex] (int_dyn) -| (sum1);
		\draw[-latex] (split_i) -- (sum4);
		\draw[-latex] (sum4) -- (meas_i)
						node[right,pos=1]{$i_m(nT_s)$};
		\draw[-latex] (noise_i) -- (sum4) 
						node[above,pos=0]{$\noii(nT_s)$};
		\draw[-latex] (noise_v) -- (sum3) 
						node[above,pos=0]{$\noiv(nT_s)$};			
	\end{tikzpicture}
	\caption{Realistic identification setup. $K=k(\vref-v)$ in
	the static static identification stage, and 
	$K=\kappa(\vref,v)$	from \eqref{eq:feedlin} in the dynamic
	stage. The block ZOH is a standard zero-order hold.}
	\label{fig:real_setup}
\end{figure}
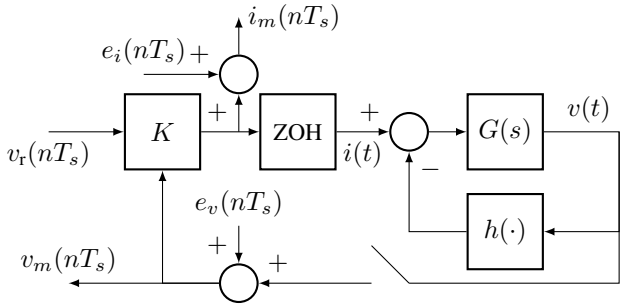

We briefly describe the simulation procedure.
Given a vector $\bar{V}_\text{r}$, each of the $M$ experiments
of the static identification stage was simulated by numerically 
integrating the dynamics of the scheme shown in Figure 
\ref{fig:real_setup}, with $K$ given by $k(\vref - v)$,
with the input $\vref(t) = \bar{V}_\text{r}[m]$, $t\ge 0$,
and with zero initial conditions\footnote{The simulations
were performed in Matlab's Simulink$^{\text{TM}}$ using the
numerical integration routine \texttt{ode15s} with a 
maximum step set to $10^4$ seconds and relative/absolute 
tolerances set to $10^{-6}$.}. Numerical integration was
carried out for $100$ seconds, which was sufficient to see
\eqref{eq:meas_estimators} converge.

To generate data for the dynamic identification stage, we 
performed $R$ simulations$^\text{4}$ corresponding to $R$ 
realizations of two periods the random-phase multisine 
inputs given by
\[
	\vref(nT_s) = \sum_{\ell=-N_f}^{N_f} u_\ell 
	\sin(\tfrac{2\pi}{N}	\ell n + \theta_\ell), \quad\quad n = 0,1,2,\dotsc
\]
where the $\theta_\ell$ are random variables uniformly 
distributed over $[0,2\pi[$, $N = T/T_s$ is the number
of samples per signal period $T$, and $N_f = f_{\max} T < N/2$
is the harmonic number corresponding to the largest
frequency in the signal, $f_{\max}$. The coefficients 
$u_\ell$ are chosen such that $u_0=0$ and $u_\ell = \bar{u}$,
with $\bar{u}$ a constant used to set the input RMS level. 
Simulations were carried out by numerically integrating the
dynamics of the scheme shown in Figure \ref{fig:real_setup}, 
with zero initial conditions, and with $K$ given by 
$\kappa(\vref,v)$ in \eqref{eq:feedlin}. 

Using the generated data, a continuous-time transfer 
function $\hat{G}_k(s)$ was estimated using the 
off-the-shelf Matlab System Identification 
Toolbox\footnote{Toolbox version 9.9, Matlab version R2018b.} 
routine \texttt{tfest}\footnote{The function \texttt{tfest} was 
used with standard settings. 
The routine initializes parameters through the Instrument 
Variable (IV) method, and updates the parameters by minimizing 
a weighted prediction error norm using a nonlinear least-squares 
search method.}. The number of poles and zeros of the identified
transfer were constrained to be the same as those of the
ground truth ones. The identified linear model is recovered
as in \eqref{eq:original_G}.

The results to be discussed next were obtained with data
generated using the parameters in Table \ref{tab:simulations}.
The signal-to-noise ratio (SNR) value refers to ratio of the
average power of the output of the noiseless system in the
dynamic identification stage, and the noise variance $\sigma^2$.

\subsection{Fitzhugh-Nagumo circuit}

Using the basis functions $\phi_j(v) = v^j$, $j=2,3$, 
Figure \ref{fig:fhn_IV} shows that assuming a realistic
setting results in a small error $(\Iinf-\Iinfh)(v)$. The error 
remains roughly the same when the noise variance is 
increased by a factor of $10$. 

Figure \ref{fig:fhn_validation} shows validation of the
identified model in closed-loop. For validation purposes, the
mean of the input $i(t)$ was set to $-1.5$, which puts the 
FHN system in the excitable regime, and results in a
characteristic spiking behavior. It can be seen that the error 
is kept low for most of the time, except at moments when the
model ``misses'' a spike. These misses occur due to the
ultrasensitivity of excitable systems with respect to their
inputs.
 
 \begin{table}
\centering
\caption{Parameters used in the generation of data.}
\label{tab:simulations} 
\begin{tabular}{|c|c|c|c|c|c|c|c|}
\hline
$\;$ & $T_s$ & k &$ f_{\max} $ & $R$ & $T$  & $\sigma$ & SNR  \\
\hline 
FHN & $10^{-3}$ s & $1.5$ & 100 Hz & $5$ & $500$ s & 0.01 & 
$40$ dB \\ 
\hline 
Chua & $10^{-3}$ s & $5$ & 100 Hz&$5$ & $500$ s & 0.01 & $40$ dB \\ 
\hline 
\end{tabular} 
\end{table}

\begin{figure}
	\centering
	\psfrag{voltage}[cc]{\small $v [V]$}
	\psfrag{current}{\small $\Iinf [A]$}
	\psfrag{error}[cb]{\small $\Iinf-\Iinfh$}
	\includegraphics[scale=0.5]{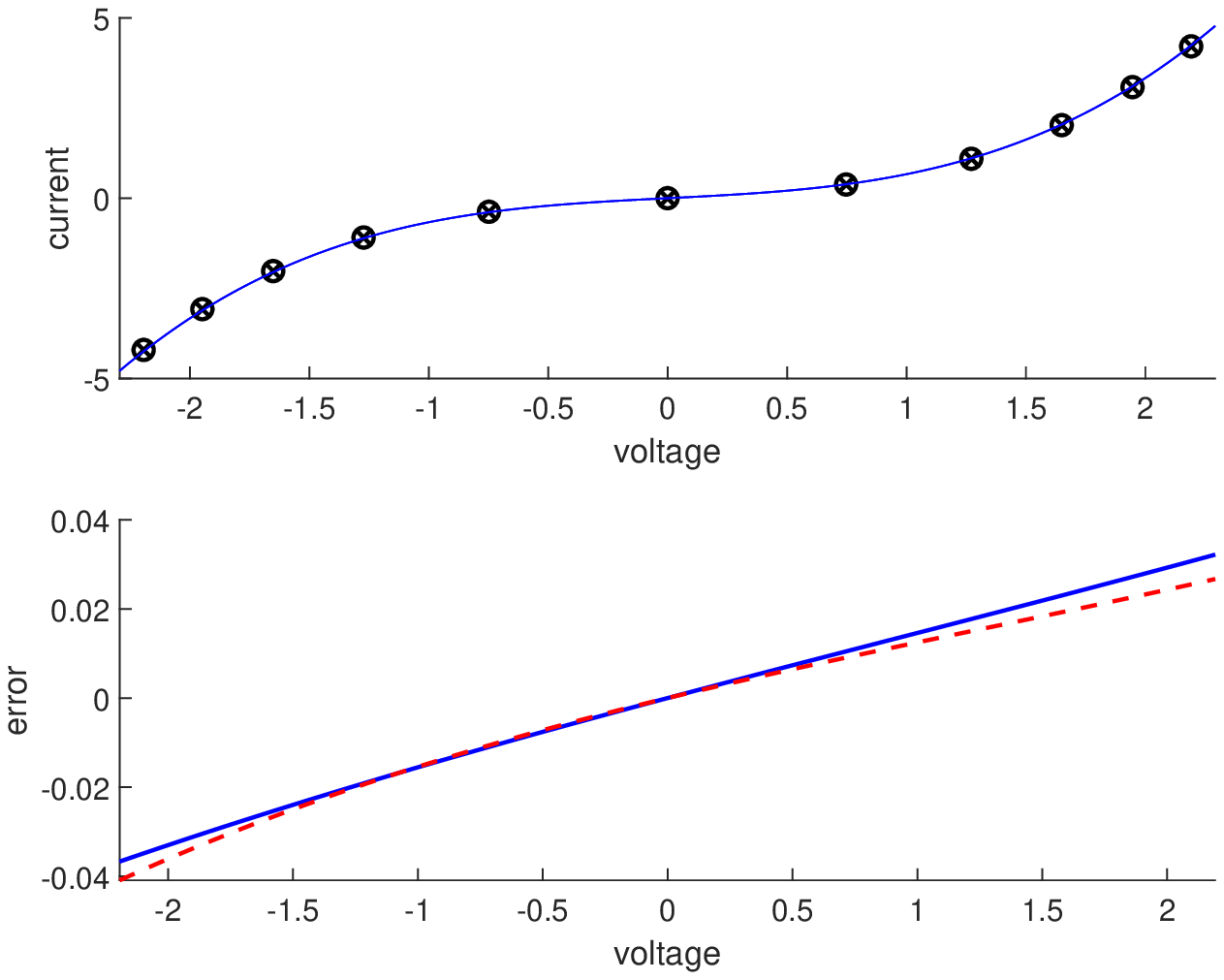}
	\caption{Estimation of $\Iinf$ for the 
	FHN	circuit. Top: ground truth 
	$\Iinf$ (line), estimates $(\hat{v},\hat{\imath})$ 
	with $\sigma = 0.01$ (crosses) and 
	$\sigma = 0.1$ (circles). Bottom: error 
	$\Iinf-\Iinfh$ with $\sigma = 0.01$
	(solid) and $\sigma = 0.1$ (dashed).}
	\label{fig:fhn_IV}
\end{figure}

\begin{figure}
	\centering
	\psfrag{volts}[bc]{\small $[V]$}
	\psfrag{current}[bl]{\small $\inp\;[A]$}
	\psfrag{error}[bl]{\small Error}
	\psfrag{time}[cc]{\small $t\;[s]$}
	\includegraphics[scale=0.5]{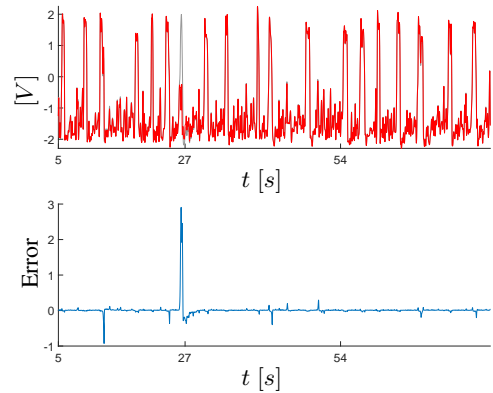}
	\caption{Validation of the identified FHN circuit.
	Top: Ground truth model output $v(t)$ (gray) and
	identified model output $\hat{v}(t)$ (red). Middle: 
	Output error $v(t) - \hat{v}(t)$. NRMSE $\approx0.84$ 
	for	the interval shown.	The NRMSE increases to about
	$0.97$ when only data from $t\ge 30$ is taken into 
	account: most of the error comes from the 
	``missed spike'' around $27$ s.}
	\label{fig:fhn_validation}
\end{figure}

\subsection{Chua's circuit}
	To 	capture the nonlinear components of a piecewise-linear 
	nonlinearity such as \eqref{eq:chua_diode},	we use the 
	basis functions $\phi_2(v) = \max\{0,v-1\}$ and 
	$\phi_3(v) = \max\{0,-(v+1)\}$. Figure \ref{fig:chua_IV}
	shows the resulting nonlinearity estimation error. Again, 
	a tenfold increase in measurement noise does not severely
	affect the error.
	
	
	
	\begin{figure}
	\centering
	\psfrag{voltage}[cc]{\small $v [V]$}
	\psfrag{current}{\small $\Iinf [A]$}
	\psfrag{error}[cb]{\small $\Iinf-\Iinfh$}
	\includegraphics[scale=0.5]{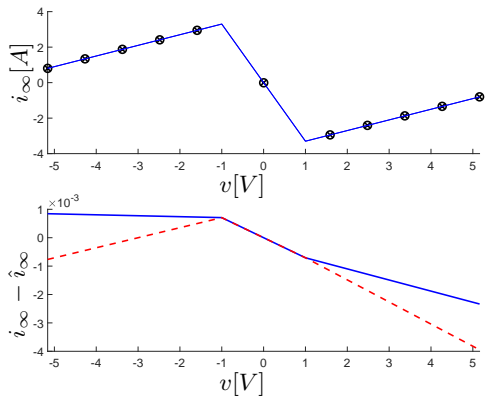}
	\caption{Estimation of $\Iinf$ for the 
	Chua circuit. Top: ground truth 
	$\Iinf$ (line), estimates $(\hat{v},\hat{\imath})$  
	with $\sigma = 0.01$ (crosses) and $\sigma = 0.1$ 
	(circles). Bottom: error $\Iinf-\Iinfh$ with 
	$\sigma = 0.01$	(solid) and $\sigma = 0.1$ (dashed).}
	\label{fig:chua_IV}
	\end{figure}
	
	\begin{figure}
	\centering
	\psfrag{volts}[bc]{\small $v\;[V]$}
	\psfrag{current}[bl]{\small $\inp\;[A]$}
	\psfrag{time}[cc]{\small $t\;[s]$}
	\includegraphics[scale=0.4]{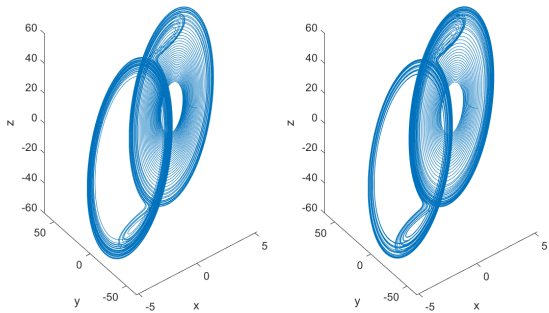}
	\caption{Attractors of the ground truth Chua circuit
	(left) and of the identified Chua circuit (right).
	The trajectories in the states $(x,y,z)$ are obtained 
	with a modal canonical state-space realization of the
	original $G_a(s) = G(s)/(1- 4 G(s))$ and of the 
	estimated $\hat{G}_a(s)$.}
	\label{fig:chua_validation}
	\end{figure}
	
	To compare the complete identified model with the ground
	truth model, we first realize the linear dynamics of each
	system (lumped with the linear component of $h$) in 
	the modal canonical state-space form. Starting from a
	nonzero initial condition, the resulting trajectories are
	shown in Figure \ref{fig:chua_validation}. It can be seen
	that the ``double-scroll'' attractors are qualitatively
	very similar.
	
	\subsection{Discussion}

	It can be argued that the choice of the feedback gain
	$k$ is key to the success of the identification 
	procedure developed in Section \ref{sec:feedlin} when it
	is applied to the more realistic case dealt with in this
	section. In principle, $k$ does not need to exceed 
	$|\rho_1|$ by a very large margin, and indeed we chose it to
	be only slightly larger	than $|\rho_1|$ in both simulations
	above. Choosing a suitable $k$ in this case can be viewed as
	part of	experiment design. While our choices were 
	good enough to
	avoid issues with the measurement noise that is fed back into
	the system, it is clear	that difficulties might arise for
	systems with a large $|\rho_1|$. If that is the case, and 
	if it is possible, analog feedback should be used.
	
\section{Conclusion and Future work}
\label{sec:conductance}

In this paper, we have observed that feedback can
simplify the identification of a nonlinear system.
We have illustrated this idea with the elementary
situation where the original system is the feedback
interconnection of a passive LTI system and a static
nonlinearity. In this case, the use of output feedback
as part of experiment design provides a straighforward
solution to the problem. This procedure is sufficient to
identify nonlinear behaviors such as excitability 
(Fitzugh-Nagumo) or chaos (Chua). 

It is important to 
mention that this method can be used as a means to obtain
initial estimates for a final identification stage
\cite{schoukens_identification_2017}, where we perform
nonlinear optimization of the simulation error of the
nonlinear feedback system. In this stage, consistency
guarantees can be obtained depending on the noise setting.

In future research, we aim to generalize the method to 
neuronal conductance-based models 
\cite{hodgkin_quantitative_1952}, in which case the fading 
memory element is dynamic rather than static.

\bibliographystyle{ieeetrans}
\bibliography{bibliography}
    
\end{document}